\documentclass[9pt,journal]{IEEEtran}

\usepackage{epsfig}
\usepackage{amssymb}
\usepackage{amsmath}
\usepackage{color}
\usepackage{amsfonts}
\newtheorem{theorem}{Theorem}

\newtheorem{definition}{Definition}
\newtheorem{lemma}{Lemma}

\newtheorem{remark}{Remark}

\newtheorem{assumption}{Assumption}

\newtheorem{proof}{Proof}
\newtheorem{proof of Theorem 1}{Proof of Theorem 1}

\begin{document}

\title{Convergence Analysis of Continuous-Time Distributed\\
Stochastic Gradient Algorithms}

\author{Jianhua Sun, Kaihong Lu, and Xin Yu 
\thanks{ }
\IEEEcompsocitemizethanks{
	\IEEEcompsocthanksitem Corresponding author: Kaihong Lu
	\IEEEcompsocthanksitem J. Sun and K. Lu are with the College of Electrical Engineering and Automation, Shandong University of Science and Technology, Qingdao 266590, China (e-mail: Sunjhsun@163.com; khong{\_}lu@163.com).

X. Yu is with the School of Electrical Engineering and Automation,
Jiangsu Normal University, Xuzhou 221116, China (e-mail: yuxin218@163.com).}
}
\maketitle
\begin{abstract}

In this paper, we propose a new framework to study distributed optimization problems with stochastic gradients by employing a multi-agent system with continuous-time dynamics. Here the goal of the agents is to cooperatively minimize the sum of convex objective functions. When making decisions, each agent only has access to a stochastic gradient of its own objective function rather than the real gradient, and can  exchange local state information with its immediate neighbors via a time-varying directed graph.  Particularly, the stochasticity is depicted by the Brownian motion. To handle this problem, we propose a continuous-time distributed stochastic gradient algorithm based on the consensus algorithm and the gradient descent strategy. Under mild assumptions on the connectivity of the graph and objective functions, using convex analysis theory, the Lyapunov theory and It\^{o} formula, we prove that the states of the agents asymptotically reach a common minimizer in expectation. Finally, a simulation example is worked out to demonstrate the effectiveness of our theoretical results.

\end{abstract}


\IEEEpeerreviewmaketitle
\section{Introduction}
\IEEEPARstart{I}{n} distributed optimization, the goal of the agents is to cooperatively minimize a global objective function that formed by the sum of local objective functions \cite{A1}. Distributed optimization has attracted increasing attention in recent years \cite{B1}-\cite{B4}. This is due to its good performance such as flexibility, strong robustness and low communication costs, as well as its wide applications in many areas such as distributed resource allocation \cite{C1}, machine learning \cite{C2} and visual human localization \cite{C3}.

Recently, considering that distributed optimization occurs in uncertain environments where real gradient information is not available, various results on distributed stochastic optimization algorithms have been achieved. For example, for distributed convex optimization, a distributed stochastic gradient tracking algorithm is proposed in \cite{C5}, where exponential convergence in expectation is achieved under a constant step size. With strongly convex functions considered, distributed stochastic gradient descent algorithms are presented in \cite{C7}, \cite{C8}, where the convergence in expectation is analyzed. For the case with non-convex objective functions, the convergence in high probability of the distributed stochastic gradient descent algorithm is achieved in \cite{A3}. For distributed optimization problems with non-convex objective functions, a distributed stochastic algorithm combining the gradient tracking algorithm and the variance reduction strategy is proposed in \cite{C6}. The case of coupled constraints is studied in \cite{C9}, where a distributed stochastic coupled diffusion algorithm is developed. In addition, distributed optimization problems with global inequality constraints are investigated in \cite{E1} using a class of stochastic gradient algorithms, while constrained optimization problems subject to time-varying communication noise are addressed in \cite{E2} via a distributed stochastic composite mirror descent method.

All the aforementioned investigations focus on discrete-time stochastic optimization algorithms. Compared with discrete-time algorithms, continuous-time algorithms activate more powerful analysis tools in control theory. In many engineering applications, the systems often possess continuous-time dynamics. For instance, in the cooperative delivery problem of unmanned ground vehicles, the dynamics of each vehicle is continuously changing and is typically modeled as a continuous-time system, where the goal of vehicles is to minimize the operational cost \cite{A2}. In fact, some new progress has been made in continuous-time distributed optimization algorithms. For example, the continuous-time distributed optimization of a sum of convex functions over directed graphs is studied in \cite{D1}. In \cite{C10}, two novel distributed continuous-time algorithms based on the subgradient method are proposed for solving constrained convex optimization problem. In \cite{C11}, a distributed subgradient-based algorithm is proposed for continuous-time multi-agent systems to find a feasible solution to convex inequality. A continuous-time algorithm for solving distributed, non-smooth, and pseudoconvex optimization problems with local convex inequality constraints is developed in \cite{C14}. A distributed Support Vector Machine method to solve the binary classification problem via dynamic balanced directed graphs is proposed in \cite{C13}. Moreover, to ensure zero constraint violation, a continuous-time distributed optimization algorithm is proposed in \cite{C15}.

It is worth noticing that when running continuous-time distributed algorithms in \cite{D1}-\cite{C15}, the accurate gradient information is necessary. However, in many practical applications, the gradients are usually influenced by noises, and their real values are not available. For instance, in distributed tracking problems in sensor networks, distance measurements from sensors to the target are usually influenced by measurement noise \cite{C4}. In such cases, agents can only make decisions using the stochastic estimates on the gradient information. For continuous-time dynamics, the Brownian motion with unbounded variance, which is widely used to capture the volatility of stock prices in the financial field \cite{C21} and the uncertainties in the field of stochastic control systems \cite{C22}, plays a fundamental role in modeling the stochasticity of the noises. Due to the influence of the Brownian motion, the stochastic gradients are not differentiable. Thus, the existing results on deterministic continuous-time distributed algorithms in \cite{D1}-\cite{C15} are not applicable to the stochastic cases, and it is necessary to adopt new tools to analyze the convergence of continuous-time distributed stochastic optimization algorithms. These factors motivate the study of this work.

In this paper, we consider a distributed convex optimization by employing a multi-agent system with continuous-time dynamics, where each agent exchanges local state information with its immediate neighbors via a time-varying directed communication graph. Unlike the continuous-time distributed optimization with accurate gradient information studied in \cite{D1}-\cite{C15}, we study the case where the agents only have access to the stochastic  gradients of their own objective functions. Moreover, different from discrete-time cases \cite{C5}-\cite{E2} where the stochasticity in the stochastic gradients is captured by the Gaussian or sub-Gaussian noise model, here the stochasticity is modeled by the Brownian motion. To address such a problem, we propose a new continuous-time distributed stochastic gradient algorithm
based on the consensus algorithm and the gradient descent strategy. The stochastic gradients are not continuous because the Brownian motion is not differentiable. This brings difficulties in analyzing the convergence of the algorithm. The It\^{o} formula is employed to deal with the non-continuity of the stochastic gradients. By combining the convex analysis theory, graph theory, and Lyapunov stability theory, the convergence of the proposed algorithm is analyzed. The results show that, if the time-varying directed graph is balanced and the uniformly strongly
connected, then the states of the agents asymptotically reach a common minimizer in expectation. Finally, a simulation example is presented to demonstrate the effectiveness of our theoretical results.

The rest of this paper is organized as follows. In Section II, some mathematical preliminaries are introduced. In Section III, the problem to be studied is formulated. In Section IV, we state our main result, and its proof is provided in detail. In Section V, a simulation example is presented. Section VI concludes the whole paper. In Section VII, proofs of the relevant lemmas are given in detail.

\section{preliminaries}\label{se1}

\subsection{Mathematical preliminaries}
In this section, we define the mathematical notations to be used and review some basic concepts related to the It\^{o} process.

Throughout this paper, $|x|$ is used to represent the absolute value of scalar $x$. $\mathbb{R}$ denotes the set of real numbers. $\mathbb{R}^m$ and $\mathbb{R}^{m\times n}$ represent the $m$-dimensional real vector space and the $m\times n$ real matrix space, respectively. $\mathbf{1}_m$ and $\mathbf{0}_m$ are $m$-dimensional vectors whose elements are all 1 and 0, respectively. $I_n$ is the $n\times n$ identity matrix. For any two vectors $\mu\in \mathbb{R}^m$ and $\nu\in \mathbb{R}^m$, the operator $\langle\mu,\nu\rangle$ denotes the inner product of $\mu$ and $\nu$. For a matrix $A\in \mathbb{R}^{m\times n}$, $[A]_{ij}$ denotes the matrix entry in the $i^{\mathrm{th}}$ row and $j^{\mathrm{th}}$ column, and $[A]_{i\cdot}$ represents the $i^{\mathrm{th}}$ row of the matrix $A$. We use $\|\cdot\|$ to denote the standard Euclidean norm for vectors $x\in \mathbb{R}^m$ and the Frobenius norm for matrices $A\in \mathbb{R}^{m\times n}$, i.e., $\|x\|=\sqrt{x^Tx}$ or $\|A\|=\sqrt{\sum_{i=1}^m\sum_{j=1}^nA_{ij}^2}$.
For single-valued function $f(\cdot): \mathbb{R}^m\rightarrow \mathbb{R}$, we denote its gradient and its Hessian matrix by $\nabla f$ and $\nabla_2 f$, respectively. For double-valued function $g(\cdot, \cdot): \mathbb{R}^m\times \mathbb{R} \rightarrow \mathbb{R}$, we denote its gradient and its Hessian matrix by $\nabla_i g(\cdot, \cdot)$ and $\nabla_{ii} g(\cdot, \cdot)$ with respect to the $i^{\mathrm{th}}$ argument, respectively. For matrices $A$ and $B$, the Kronecker product is denoted by $A\otimes B$. $\mathbb{E}[\cdot]$ denotes the expectation of a random variable. Given two positive variables $a$ and $b$, $a=\mathcal{O}(b)$
means that the order of $a$'s bound is not larger than the order of $b$, i.e., $\mathrm{lim}_{b\rightarrow \infty}\frac{a}{b}< \infty$.
Given a set of vectors $x_i\in \mathbb{R}^m$, where $i\in \mathcal{V}$ and $\mathcal{V}=\{1,\cdots,n\}$, we denote $\mathrm{col}(x_i)_{i\in \mathcal{V}}=[x_1^T,\cdots, x_n^T]^T$. Moreover,  $\mathrm{diag}(x_i)_{i\in \mathcal{V}}$ denotes the block-diagonal matrix $\mathrm{diag}(x_1,\ldots,x_n)$.

Let $(\Omega, \mathcal{F}, \{\mathcal{F}_t\}_{t\ge0}, \mathbb{P})$ be a filtered probability space.
For $p>0$ and $0\le a<b$, an $\mathbb{R}^m$-valued process
$f=\{f(t)\}_{t\in[a,b]}$ belong to $\mathcal{L}^p([a,b],\mathbb{R}^m)$, denoted by $\{f(t)\}_{t\in[a,b]}\in \mathcal{L}^p([a,b],\mathbb{R}^m)$,
if it is $\{\mathcal{F}_t\}$-adapted and satisfies
\[\int_a^b \|f(t)\|^pdt<\infty.\]
Furthermore, for $p>0$ and $0\le a<b$, an $\mathbb{R}^m$-valued process
$f=\{f(t)\}_{t\in[a,b]}$ belong to $\mathcal{M}^p([a,b],\mathbb{R}^m)$, denoted by $\{f(t)\}_{t\in[a,b]}\in \mathcal{M}^p([a,b],\mathbb{R}^m)$,
if it is $\{\mathcal{F}_t\}$-adapted and satisfies
\[\mathbb{E}\left[\int_a^b \|f(t)\|^p\,dt\right]<\infty.\]
We write $h\in \mathcal{L}^p(\mathbb{R}_+,\mathbb{R}^m)$, with $p>0$, if $h\in \mathcal{L}^p([0,t],\mathbb{R}^m)$ for any $t>0$.

Below, we briefly review the relevant concepts of It\^{o} process and It\^{o} formula. For more detailed information, please refer to \cite{C16}-\cite{C19}. Let $B=\{B(t)\}_{t\geq0}$ be a one dimensional Brownian motion defined on our probability space and let $x=\{x(t)\}_{t\geq0}$ be an $\mathcal{F}(t)$-adapted process taking values on $\mathbb{R}^m$.

\begin{definition}\label{DE1}
Let $b\in \mathcal{L}^1(\mathbb{R}_+,\mathbb{R}^m)$ and $\sigma\in \mathcal{L}^2(\mathbb{R}_+,\mathbb{R}^m)$. A process $\{x(t)\}_{t\geq0}$ is called an It\^{o} process if
\[x(t)=x_0+\int_0^t b(s)ds+\int_0^t \sigma(s)dB(s),\]
which is equivalently written as
\begin{equation}\label{eq130}
dx(t)=b(t)dt+\sigma(t)dB(t).
\end{equation}
\end{definition}
In the following lemma, let us recall the It\^{o} formula.
\begin{lemma}[It\^{o} formula]\label{LEIT1}
Consider It\^{o} process (\ref{eq130}). Given function $\mathcal{E}(x,t): \mathbb{R}^m\times \mathbb{R}\rightarrow \mathbb{R}$, if $\mathcal{E}(x,t)$ is twice differentiable with respect to $x$ and differentiable with respect to $t$, then
\begin{equation*}\label{eqIt}
\begin{split}
d\mathcal{E}(x(t),t)&=\frac{\partial\mathcal{E}(x(t),t)}{\partial t}dt+\langle\nabla_1\mathcal{E}(x(t),t),f(t)\rangle dt \\
&~~~+\frac{1}{2}\mathrm{Tr}\left(\sigma(t)\sigma(t)^T\nabla_{11}\mathcal{E}(x(t),t)\right)dt \\
&~~~+\langle\nabla_1\mathcal{E}(x(t),t),\sigma(t)\rangle dB(t).
\end{split}
\end{equation*}
\end{lemma}

\subsection{Graph theory}
Consider a time-varying directed communication graph $\mathcal{G}(t)=(\mathcal{V},\mathcal{E}(t),\mathcal{A}(t))$, where $\mathcal{V}=\{1,\cdots,n\}$ is a set of vertices, $n$ is the number of agents, $\mathcal{E}(t)\subset \mathcal{V}\times \mathcal{V}$ is an edge set, and $\mathcal{A}(t)=(a_{ij})_{n\times n}$ is a non-negative matrix for adjacency weights of edges such that $a_{ij}(t)>0\Leftrightarrow(j,i)\in \mathcal{E}(t)$ and $a_{ij}(t)=0$ otherwise. Here we assume that $a_{ii}(t)=0$ for any $i\in \mathcal{V}$ and $t\geq0$. $\mathcal{N}_i(t)=\{j\in \mathcal{V}|(j,i)\in \mathcal{E}(t)\}$ represents the neighbor set of agent $i$ at time $t$. $\mathcal{G}(t)$ is said to be balanced if the sum of the interaction weights from and to agent $i$ are equal, i.e., $\sum_{j=1}^na_{ij}(t)=\sum_{j=1}^na_{ji}(t)$. The Laplacian matrix of the graph at time $t$ is defined as $L(t)=(l_{ij}(t))_{n\times n}$, where $l_{ii}(t)=\sum_{j=1}^na_{ij}(t)$ and $l_{ij}(t)=-a_{ij}(t), i\neq j$ for any $i,j\in \mathcal{V}$. For a fixed topology $\mathcal{G}=(\mathcal{V},\mathcal{E},\mathcal{A})$, a path of length $r$ from node $i_1$ to node $i_{r+1}$ is a sequence of $r+1$ distinct nodes $i_1,\cdots,i_{r+1}$ such that $(i_q,i_{q+1})\in \mathcal{E}$ for $q=1,\cdots,r$. If there exists a path between any two nodes, then $\mathcal{G}$ is said to be strongly connected. For time-varying graph $\mathcal{G}(t)$, $(j,i)$ is called a $(\delta,T_c)$-edge if there always exist two positive real numbers $T_c$ and $\delta$ such that $\int_t^{t+T_c}a_{ij}(s)ds\geq\delta$ for any $t>0$. A $(\delta,T_c)$-graph, corresponding to $\mathcal{G}(t)$, is defined as $\mathcal{G}_{(\delta,T_c)}=(\mathcal{V},\mathcal{E}_{(\delta,T_c)})$, where $\mathcal{E}_{(\delta,T_c)}=\{(j,i)\in \mathcal{V}\times \mathcal{V}|\int_t^{t+T_c}a_{ij}(s)ds\geq\delta, \forall t\geq0\}$. Moreover, if the $(\delta,T_c)$-graph $\mathcal{G}_{(\delta,T_c)}$ is strongly connected, we call the time-varying graph $\mathcal{G}(t)$ is $(\delta,T_c)$-strongly connected.

Throughout this paper, we make the following assumption for the communication graph.
\begin{assumption}\label{ASB1}
The digraph $\mathcal{G}(t)$ is $(\delta,T_c)$-strongly connected and balanced.
\end{assumption}

Consider the following consensus model:
\begin{equation}\label{eqb1}
\dot{\chi}(t)=-L(t)\chi(t),
\end{equation}
where $\chi(t)\in \mathbb{R}^m$ and $t\geq0$. Assumption \ref{ASB1} ensures consensus model (\ref{eqb1}) reaches average-consensus \cite{C23}, i.e., $\lim_{t \to \infty}\chi(t)=\frac{\mathbf{1}_n^T\chi(0)\mathbf{1}_n}{n}$. The state transition matrix of model (\ref{eqb1}) is defined as
\begin{equation}\label{eqb2}
\begin{split}
\Phi(t,s)&=I-\int_s^tL(\tau)d\tau\\
&~~~+\int_s^t\int_s^{\tau_1}L(\tau_1)L(\tau_2)d\tau_1d\tau_2+\cdots.
\end{split}
\end{equation}
Obviously, $\Phi(t,s)$ is a stochastic matrix. Based on Lemma 3 in Lu et al. \cite{C11}, there is an important lemma as follows.
\begin{lemma}[\cite{C11}]\label{LEb1}
Under Assumption \ref{ASB1}, for some $\mathcal{C}>0$ and $0<\lambda<1$,
\begin{equation*}
\left|[\Phi(t,s)]_{ij}-\frac{1}{n}\right|\leq \mathcal{C}\lambda^{t-s}, \forall i,j\in\{1,\cdots,n\},
\end{equation*}
where $t\geq s\geq0$.
\end{lemma}

\section{Problem formulation}
Consider a multi-agent system (MAS) consisting of $n$ agents, labeled by the set $\mathcal{V}=\{1,\cdots,n\}$, where each agent exchanges local state information with its neighbours through a time-varying directed graph $\mathcal{G}(t)$. The objective of agents is to cooperatively solve the following optimization problem
\begin{equation}\label{eqc1}
\min_{{x}\in \mathbb{R}^m} f({x}),~f({x})=\sum_{i=1}^n f_i({x})
\end{equation}
where $f_i:\mathbb{R}^m\rightarrow \mathbb{R}$ is twice continuously differentiable and convex. When making decisions, each agent $i$ only has access to a noisy gradient of $f_i$, denoted by $\nabla \tilde{f}_i(x_i(t))$, rather than its real gradient. The following assumptions on the objective functions are adopted throughout this paper.
\begin{assumption}\label{ASC1}
The gradient $\nabla f_i(x)$ is bounded for any $i\in \mathcal{V}$, i.e., there exists some constant $M>0$ such that
$$\|\nabla f_i(x)\|\leq M.$$
\end{assumption}
\begin{assumption}\label{ASC2}
Each function $f_i$ is $L$-smooth, i.e., there exists some constant $L>0$ such that
$$\|\nabla f_i(x)-\nabla f_i(y)\|\leq L\|x-y\|.$$
\end{assumption}

To solve the problem (\ref{eqc1}), we consider the following continuous-time distributed stochastic gradient algorithm:
\begin{equation}\label{eq3_1}
dx_i(t)=\sum_{j=1}^na_{ij}(t)(x_j(t)-x_i(t))dt-\eta_t\nabla \tilde{f}_i(x_i(t)),
\end{equation}
where $x_i(t)\in \mathbb{R}^m$ represents the state of agent $i$ with initial value $x_i(0)$, $\nabla \tilde{f}_i(x_i(t))=\nabla f_i(x_i(t))dt+g_i(t)dB_i(t)$, $B_i(t)$ is a Brownian motion, $g_i(t)\in \mathbb{R}^m$ is the noise intensity, and $\eta_t>0$ is a non-increasing step size.

\begin{remark}\label{REH1}
The stochastic differential equation (\ref{eq3_1}) can be viewed as the continuous-time counterpart of a discrete-time distributed stochastic gradient descent algorithm with time-varying noise intensity. In particular, applying an Euler discretization yields the update
\begin{equation}\label{eq3_2}
\begin{split}
x_i^{k+1}&=x_i^k+\sum_{j=1}^n(a_{ij}(k)(x_j^k-x_i^k))h\\
&~~~-\eta_{k}h(\nabla f_i(x_i^k)+\frac{g_i(k)\xi_i^k}{\sqrt{h}}),
\end{split}
\end{equation}
where $h>0$ is a step size. Let $\xi_i^k\triangleq B_i(k+1)-B_i(k)\sim (0, h_0\delta^2)$ and $\hat{\nabla} f_i(x_i^k)=\nabla f_i(x_i^k)+\frac{g_i(k)\xi_i^k}{\sqrt{h}}$ for some $h_0> 0$. It is easy to verify that if $g_i(k)=\frac{\mathbf{1}_m}{\sqrt{m}}$ and $h_0\leq h$, then the Gaussian model can be achieved as $\mathbb{E}[\hat{\nabla} f_i(x_i^k)-\nabla f_i(x_i^k)]=\mathbb{E}\left[\frac{g_i(k)\xi_i^k}{\sqrt{h}}\right]=\mathbf{0}_m$ and $\mathbb{E}[\|\hat{\nabla} f_i(x_i^k)-\nabla f_i(x_i^k)\|^2]=\mathbb{E}\left[\left\|\frac{g_i(k)\xi_i^k}{\sqrt{h}}\right\|^2\right]\leq\delta^2$.
The two conditions are commonly adopted in the study of discrete-time distributed stochastic optimization \cite{C5}, \cite{A3}. Consequently, (\ref{eq3_1}) is the continuous-time counterpart of the discrete-time model in (\ref{eq3_2}). Studying the convergence of (\ref{eq3_1}) needs some powerful analysis tools for stochastic differential equations due to the influence of the Brownian motion $B(t)$. To the best of our knowledge, this is the first time to study the continuous-time distributed stochastic optimization algorithm.
\end{remark}

To ensure the convergence of algorithm (\ref{eq3_1}), the following assumption on the noise intensity is necessary.
\begin{assumption}\label{AS3-1}
For any $i\in \mathcal{V}$, $\|g_i(t)\|\leq K$ for some $K>0$.
\end{assumption}

In the next section, we will present the main result of this paper and provide its proof.

\section{Main results}

Let us start this section by stating our main result in the following theorem.
\begin{theorem}\label{TH1}
Under Assumptions~\ref{ASB1}--\ref{AS3-1}, if
$\eta_s=\frac{\beta}{(s+1)^a}$, then for some $\frac{1}{2}<a\le1$ and $\beta>0$, by algorithm (\ref{eq3_1}), for any $i\in\mathcal{V}$,
\begin{equation}\label{eqco1}
\mathbb{E}[f(x_i(t))-f(x^*)]\le
\begin{cases}
\mathcal{O}\left(t^{-(a-\frac{1}{2})}\right), & \frac{1}{2}<a<\frac{3}{4},\\
\mathcal{O}\left(\dfrac{\sqrt{\ln t}}{t^{1/4}}\right), & a=\frac{3}{4},\\
\mathcal{O}\left(t^{-(1-a)}\right), & \frac{3}{4}<a<1,\\
\mathcal{O}\left((\ln t)^{-1}\right), & a=1,
\end{cases}
\end{equation}
where $x^*\in \mathbb{R}^m$ is a minimizer of (\ref{eqc1}).
\end{theorem}

The result in Theorem \ref{TH1} shows that, under algorithm (\ref{eq3_1}), the states of the agents asymptotically converge to a common minimizer of (\ref{eqc1}) in expectation. Moreover, the convergence rate depends on the decay exponent $a$ of the step size $\eta_t$. It is obvious that the fastest polynomial convergence rate is achieved if $a=\frac{3}{4}$. Before proving Theorem \ref{TH1}, some necessary lemmas are provided.
\begin{lemma}[\cite{C16}]\label{LEE1}
If $p\geq 2$, $g(t)\in \mathcal{M}^p([0,t],\mathbb{R}^{d\times m})$, then
\begin{equation*}
\begin{split}
&\mathbb{E}\left[\left\|\int_0^tg(s)dB(s)\right\|^p\right]\\
&\leq\left(\frac{p(p-1)}{2}\right)^\frac{p}{2}t^\frac{p-2}{2}\mathbb{E}\left[\int_0^t\|g(s)\|^pds\right].
\end{split}
\end{equation*}
In particular, for $p=2$, the equation holds.
\end{lemma}
\begin{lemma}\label{LE3}
For any $a>0$ and $0<\lambda<1$,
\begin{equation*}\label{eq112}
\int_0^t\frac{\lambda^{-s}}{(s+1)^a}ds\leq\frac{\delta_1\lambda^{-t}}{(t+1)^a}+\frac{\delta_2}{(t+1)^a}+\delta_3,
\end{equation*}
where $\delta_1=\frac{t_0+1}{\ln\lambda^{-1}(t_0+1)-a}$, $\delta_2=\lambda^{-t_0}t_0(t_0+1)^a$, $\delta_3=\frac{a\lambda^{-t_0}t_0(t_0+1)}{\ln\lambda^{-1}(t_0+1)-a}$ and $t_0>\max(\frac{a}{\ln\lambda^{-1}}-1,0)$.
\end{lemma}
\begin{proof}
See APPENDIX. \emph{A}.
\end{proof}

Then the bounds on $\|x_i(t)-\bar{x}(t)\|$ and $\int_0^t\phi(s)\|x_i(s)-\bar{x}(s)\|ds$ are provided, refer to the following lemma for details.
\begin{lemma}\label{LE1}
Under Assumptions \ref{ASB1} and \ref{ASC1}, if $\phi(t)>0$ is a non-increasing with respect to $t$, then by algorithm (\ref{eq3_1}), for any $t\geq0$ and $i\in\mathcal{V}$,
\begin{equation}\label{eq3-1}
\|x_i(t)-\bar{x}(t)\|\leq \theta_1\lambda^{t}+\theta_2\int_0^t\lambda^{t-s}\eta_s ds+\|\Psi(t)\|
\end{equation}
and
\begin{equation}\label{eq3-2}
\begin{split}
&\int_0^t\phi(s)\|x_i(s)-\bar{x}(s)\|ds\\
&\leq \frac{\theta_1\phi(0)}{\ln\lambda^{-1}}(1-\lambda^t)+\frac{\theta_2}{\ln\lambda^{-1}}(1-\lambda^t)\int_0^t\phi(s)\eta_sds\\
&~~~+\int_0^t\phi(s)\|\Psi(s)\|ds.
\end{split}
\end{equation}
where $\Psi(t)=\int_0^t\left(\left([\Phi(t,\tau)]_{i\cdot}-\frac{\mathbf{1}_n^T}{n}\right)\otimes I_m\right)\eta_\tau g(\tau)dB(\tau)$,
$\theta_1=\sqrt{nm}\mathcal{C}\|{x}(0)\|$, $\theta_2=\sqrt{m}n^{\frac{3}{2}}\mathcal{C}M$, $\bar{x}(t)=\frac{1}{n}\sum_{i=1}^nx_i(t)$,
and $\mathcal{C}$, $\lambda$ are defined in Lemma \ref{LEb1}.
\end{lemma}
\begin{proof}
See APPENDIX. \emph{B}.
\end{proof}

In what follows, we further establish the bounds on the expectations of $\|x_i(t)-\bar{x}(t)\|$ and $\int_0^t\phi(s)\|x_i(s)-\bar{x}(s)\|ds$.
\begin{lemma}\label{LE2}
Under Assumptions \ref{ASB1}, \ref{ASC1} and \ref{AS3-1}, if $\phi(t)>0$ is a non-increasing with respect to $t$, then for any $t\geq0$ and $i\in\mathcal{V}$,
\begin{equation}\label{eq3-h1}
\begin{split}
\mathbb{E}[\|x_i(t)-\bar{x}(t)\|]&\leq \theta_1\lambda^t+\theta_2\int_0^t\lambda^{t-s}\eta_s ds\\
&~~~+\theta_3\left(\int_0^t(\lambda^2)^{t-s}\eta_s^2ds\right)^{\frac{1}{2}}
\end{split}
\end{equation}
and
\begin{equation}\label{eq3-8}
\begin{split}
&\mathbb{E}\left[\int_0^t\phi(s)\|x_i(s)-\bar{x}(s)\|ds\right]\\
&\leq\frac{\theta_1\phi(0)}{\ln\lambda^{-1}}(1-\lambda^t)+\frac{\theta_2}{\ln\lambda^{-1}}(1-\lambda^t)\int_0^t\phi(s)\eta_sds\\
&~~~+\frac{\theta_3}{\sqrt{\ln\lambda^{-2}}}\sqrt{1-\lambda^{2t}}\left(\int_0^t\phi(s)^2ds\int_0^t\eta_s^2ds\right)^{\frac{1}{2}}.
\end{split}
\end{equation}
where $\theta_1$, $\theta_2$ and $\bar{x}$ are defined in Lemma \ref{LE1}, and $\theta_3=\sqrt{m}n^{\frac{3}{2}}\mathcal{C}K$.
\end{lemma}
\begin{proof}
See APPENDIX. \emph{C}.
\end{proof}

Using Lemma \ref{LE3} to (\ref{eq3-h1}) in Lemma \ref{LE2}, it is not difficult to verify that $\mathrm{lim}_{t\rightarrow \infty}\mathbb{E}[\|x_i(t)-\bar{x}(t)\|]=0$, i.e., the states of the agents asymptotically reach a common point in expectation. Based on the achieved lemmas, we are now ready to prove Theorem \ref{TH1}.
\begin{proof}[\textbf{Proof of Theorem~1}]
Consider a Lyapunov function as follows:
\begin{equation*}
V(\bar{x}(t),t)=\frac{1}{n}\varphi(t)(f(\bar{x}(t))-f(x^*))+\frac{1}{2}\|\bar{x}(t)-x^*\|^2,
\end{equation*}
where $\varphi(t)=\int_0^t\eta_sds$.
From algorithm (\ref{eq3_1}), it follows that
\begin{equation}\label{eq111}
d\bar{x}(t)=-\frac{1}{n}\eta_t\left(\sum_{i=1}^n\nabla f_i(x_i(t))dt+\sum_{i=1}^ng_i(t)dB_i(t)\right).
\end{equation}
By defining $\tilde{g}(t)=\left(g_1(t), g_2(t), \ldots, g_n(t)\right)$ and $B(t)=\mathrm{col}(B_i(t))_{i\in \mathcal{V}}$, (\ref{eq111}) is equivalent to
\begin{equation*}
d\bar{x}(t)=-\frac{1}{n}\eta_t\left(\sum_{i=1}^n\nabla f_i(x_i(t))dt+\tilde{g}(t)dB(t)\right).
\end{equation*}
Using Lemma \ref{LEIT1}, we have
\begin{equation}\begin{split}\label{eq3-11}
&dV(\bar{x}(t),t)\\
&=\frac{\partial V(\bar{x}(t),t)}{\partial t}dt+\left\langle\nabla_1 V(\bar{x}(t),t),-\frac{1}{n}\eta_t\sum_{i=1}^n\nabla f_i(x_i(t))\right\rangle dt\\
&~~~+\frac{1}{2}\mathrm{Tr}\left(\frac{1}{n^2}\eta_t^2\tilde{g}(t)\tilde{g}(t)^T\nabla_{11}V(\bar{x}(t),t)\right)dt\\
&~~~+\left\langle\nabla_1V(\bar{x}(t),t),-\frac{1}{n}\eta_t\tilde{g}(t)\right\rangle dB(t)\\
&=\frac{1}{2n^2}\eta_t^2\mathrm{Tr}\left(\tilde{g}(t)\tilde{g}(t)^T\left(\frac{1}{n}\varphi(t)\nabla_2f(\bar{x}(t))+I_m\right)\right)dt\\
&~~~+\left\langle\frac{1}{n}\varphi(t)\nabla f(\bar{x}(t))+(\bar{x}(t)-x^*),-\frac{1}{n}\eta_t\tilde{g}(t)\right\rangle dB(t)\\
&~~~+\frac{1}{n^2}\varphi(t)\eta_t\left\langle\nabla f(\bar{x}(t)),-\sum_{i=1}^n\nabla f_i(x_i(t))\right\rangle dt\\
&~~~+\frac{1}{n}\eta_t\left\langle \bar{x}(t)-x^*,-\sum_{i=1}^n\nabla f_i(x_i(t))\right\rangle dt\\
&~~~+\frac{1}{n}\eta_t(f(\bar{x}(t))-f(x^*))dt.
\end{split}\end{equation}
Note that
\begin{equation}\begin{split}\label{eq3-12}
&\left\langle\nabla f(\bar{x}(t)),-\sum_{i=1}^n\nabla f_i(x_i(t))\right\rangle\\
&=-\left\langle\nabla f(\bar{x}(t)),\sum_{i=1}^n(\nabla f_i(x_i(t))-\nabla f_i(\bar{x}(t)))\right\rangle\\
&~~~-\left\langle\nabla f(\bar{x}(t)),\sum_{i=1}^n\nabla f_i(\bar{x}(t))\right\rangle\\
&\leq\sum_{i=1}^n\|\nabla f(\bar{x}(t))\|\|\nabla f_i(x_i(t))-\nabla f_i(\bar{x}(t))\|\\
&~~~-\|\nabla f(\bar{x}(t))\|^2\\
&\leq nML\sum_{i=1}^n\|x_i(t)-\bar{x}(t)\|,
\end{split}\end{equation}
where the first inequality is due to the Cauchy-Schwarz inequality, and the last inequality is obtained by discarding a negative term and using Assumptions \ref{ASC1} and \ref{ASC2}. Furthermore,
\begin{equation}\begin{split}\label{eq3-13}
&\left\langle\bar{x}(t)-x^*,-\sum_{i=1}^n\nabla f_i(x_i(t))\right\rangle\\
&=\sum_{i=1}^n\langle x^*-x_i(t),\nabla f_i(x_i(t))\rangle\\
&~~~+\sum_{i=1}^n\langle x_i(t)-\bar{x}(t),\nabla f_i(x_i(t))\rangle\\
&\leq\sum_{i=1}^n(f_i(x^*)-f_i(x_i(t)))+M\sum_{i=1}^n\|x_i(t)-\bar{x}(t)\|\\
&=f(x^*)-\sum_{i=1}^nf_i(x_i(t))+M\sum_{i=1}^n\|x_i(t)-\bar{x}(t)\|,
\end{split}\end{equation}
where the inequality follows from the convexity of the function and the Cauchy-Schwarz inequality. Note that
\begin{equation}\begin{split}\label{eq3-14}
&\mathrm{Tr}\left(\tilde{g}(t)\tilde{g}(t)^T\left(\frac{1}{n}\varphi(t)\nabla_2 f(\bar{x}(t))+I_m\right)\right)\\
&=\frac{1}{n}\varphi(t)\mathrm{Tr}\left(\tilde{g}(t)\tilde{g}(t)^T\nabla_2 f(\bar{x}(t))\right)+\mathrm{Tr}\left(\tilde{g}(t)\tilde{g}(t)^T\right)\\
&\leq(\varphi(t)L+1)mnK^2,
\end{split}\end{equation}
where the last inequality follows from the fact that $\mathrm{Tr}(PQ)\leq m\|P\|_S\|Q\|_S$, for any $m\times m$ matrices $P$ and $Q$, and $\|\cdot\|_S$ denotes the spectral norm. Substituting (\ref{eq3-12})-(\ref{eq3-14}) into (\ref{eq3-11}) yields
\begin{equation*}\begin{split}
&dV(\bar{x}(t),t)\\
&\leq\frac{\eta_t}{n}(f(\bar{x}(t))-f(x^*))dt+\frac{ML\varphi(t)\eta_t}{n}\sum_{i=1}^n\|x_i(t)-\bar{x}(t)\|dt\\
&~~~+\frac{\eta_t}{n}\left(f(x^*)-\sum_{i=1}^nf_i(x_i(t))+M\sum_{i=1}^n\|x_i(t)-\bar{x}(t)\|\right)dt\\
&~~~+\left\langle\frac{1}{n}\varphi(t)\nabla f(\bar{x}(t))+(\bar{x}(t)-x^*),-\frac{1}{n}\eta_t\tilde{g}(t)\right\rangle dB(t)\\
&~~~+\frac{\eta_t^2mK^2}{2n}(\varphi(t)L+1)dt\\
&\leq\frac{\eta_t}{n}\sum_{i=1}^n\left|f_i(\bar{x}(t))-f_i(x_i(t))\right|dt+\frac{\eta_t^2mK^2}{2n}(\varphi(t)L+1)dt\\
&~~~+\left\langle\frac{1}{n}\varphi(t)\nabla f(\bar{x}(t))+(\bar{x}(t)-x^*),-\frac{1}{n}\eta_t\tilde{g}(t)\right\rangle dB(t)\\
&~~~+\left(\frac{ML\varphi(t)}{n}+\frac{M}{n}\right)\sum_{i=1}^n\eta_t\|x_i(t)-\bar{x}(t)\|dt\\
\end{split}\end{equation*}
\begin{equation*}\begin{split}
&\leq\left\langle\frac{1}{n}\varphi(t)\nabla f(\bar{x}(t))+(\bar{x}(t)-x^*),-\frac{1}{n}\eta_t\tilde{g}(t)\right\rangle dB(t)\\
&~~~+\left(\frac{ML\varphi(t)}{n}+\frac{2M}{n}\right)\sum_{i=1}^n\eta_t\|x_i(t)-\bar{x}(t)\|dt\\
&~~~+\frac{\eta_t^2mK^2}{2n}(\varphi(t)L+1)dt.
\end{split}\end{equation*}
Moreover, by the definition of the It\^{o} process, we know that for any $t>0$,
\begin{equation}\begin{split}\label{eq1112}
&V(\bar{x}(t),t)-V(\bar{x}(0),0)\\
&\leq\int_0^t\left\langle\frac{1}{n}\varphi(s)\nabla f(\bar{x}(s))+(\bar{x}(s)-x^*),-\frac{1}{n}\eta_s\tilde{g}(s)\right\rangle dB(s)\\
&~~~+\sum_{i=1}^n\frac{ML}{n}\int_0^t\varphi(s)\eta_s\|x_i(s)-\bar{x}(s)\|ds\\
&~~~+\sum_{i=1}^n\frac{2M}{n}\int_0^t\eta_s\|x_i(s)-\bar{x}(s)\|ds\\
&~~~+\frac{mK^2}{2n}\int_0^t\eta_s^2(\varphi(s)L+1)ds.
\end{split}\end{equation}
Taking expectations of both sides of (\ref{eq1112}) and using the martingale property of It\^{o} integral, there holds
\begin{equation*}\begin{split}
&\mathbb{E}[V(\bar{x}(t),t)]-V(\bar{x}(0),0)\\
&\leq\sum_{i=1}^n\frac{ML}{n}\mathbb{E}\left[\int_0^t\varphi(s)\eta_s\|x_i(s)-\bar{x}(s)\|ds\right]\\
&~~~+\sum_{i=1}^n\frac{2M}{n}\mathbb{E}\left[\int_0^t\eta_s\|x_i(s)-\bar{x}(s)\|ds\right]\\
&~~~+\frac{mK^2}{2n}\int_0^t\eta_s^2(\varphi(s)L+1)ds.
\end{split}\end{equation*}
Based on the definition of $V(\bar{x}(t),t)$, we have
\begin{equation}\begin{split}\label{eq3-15}
&\frac{1}{n}\varphi(t)\mathbb{E}[f(\bar{x}(t))-f(x^*)]+\frac{1}{2}\mathbb{E}[\|\bar{x}(t)-x^*\|^2]\\
&\leq\frac{1}{2}\|\bar{x}(0)-x^*\|^2+\frac{mK^2}{2n}\int_0^t\eta_s^2(\varphi(s)L+1)ds\\
&~~~+\sum_{i=1}^n\frac{ML}{n}\mathbb{E}\left[\int_0^t\varphi(s)\eta_s\|x_i(s)-\bar{x}(s)\|ds\right]\\
&~~~+\sum_{i=1}^n\frac{2M}{n}\mathbb{E}\left[\int_0^t\eta_s\|x_i(s)-\bar{x}(s)\|ds\right].
\end{split}\end{equation}
By (\ref{eq3-15}), we have
\begin{equation}\begin{split}\label{eq3-16}
&\mathbb{E}[f(\bar{x}(t))-f(x^*)]\\
&\leq\frac{n}{2\varphi(t)}\|\bar{x}(0)-x^*\|^2+\frac{mK^2}{2\varphi(t)}\int_0^t\eta_s^2(\varphi(s)L+1)ds\\
&~~~+\frac{ML}{\varphi(t)}\sum_{i=1}^n\mathbb{E}\left[\int_0^t\varphi(s)\eta_s\|x_i(s)-\bar{x}(s)\|ds\right]\\
&~~~+\frac{2M}{\varphi(t)}\sum_{i=1}^n\mathbb{E}\left[\int_0^t\eta_s\|x_i(s)-\bar{x}(s)\|ds\right].
\end{split}\end{equation}
Note that $\varphi(s)\eta_s=\eta_s\int_0^s\eta_\tau d\tau=\frac{\beta^2}{1-a}\left((s+1)^{1-2a}-(s+1)^{-a}\right)$, so $\varphi(s)\eta_s$ is non-increasing if $a>\frac{1}{2}$. Based on Assumption \ref{ASC1}, we know that for any $t>0$,
\begin{equation}\begin{split}\label{eq3-17}
&f(x_i(t))-f(x^*)\\
&=(f(x_i(t))-f(\bar{x}(t)))+(f(\bar{x}(t))-f(x^*))\\
&\leq\sum_{j=1}^n|f_j(x_i(t))-f_j(\bar{x}(t))|+(f(\bar{x}(t))-f(x^*))\\
&\leq\sum_{j=1}^nM\|x_i(t)-\bar{x}(t)\|+(f(\bar{x}(t))-f(x^*))\\
&=nM\|x_i(t)-\bar{x}(t)\|+(f(\bar{x}(t))-f(x^*)).
\end{split}\end{equation}
Taking expectations on both sides of (\ref{eq3-17}) yields
\begin{equation}\begin{split}\label{eq3-18}
\mathbb{E}[f(x_i(t))-f(x^*)]&\leq nM\mathbb{E}[\|x_i(t)-\bar{x}(t)\|]\\
&~~~+\mathbb{E}[f(\bar{x}(t))-f(x^*)].
\end{split}\end{equation}
Substituting (\ref{eq3-16}) into (\ref{eq3-18}), and using Lemma \ref{LE2}, there holds
\begin{equation}\label{eqth1}
\begin{split}
&\mathbb{E}[f(x_i(t))-f(x^*)]\\
&\leq nM\left(\theta_1\lambda^t+\theta_2\int_0^t\lambda^{t-\tau}\eta_\tau d\tau+\theta_3\left(\int_0^t(\lambda^2)^{t-\tau}\eta_\tau^2d\tau\right)^{\frac{1}{2}}\right)\\
&~~~+\left(\frac{nML\theta_2(1-\lambda^t)}{\ln\lambda^{-1}}+\frac{mLK^2}{2}\right)\frac{1}{\varphi(t)}\int_0^t\varphi(s)\eta_s^2ds\\
&~~~+\frac{nML\theta_3\sqrt{1-\lambda^{2t}}}{\sqrt{\ln\lambda^{-2}}}\frac{1}{\varphi(t)}\left(\int_0^t\varphi(s)^2\eta_s^2ds\int_0^t\eta_s^2ds\right)^{\frac{1}{2}}\\
&~~~+\frac{n}{2\varphi(t)}\|\bar{x}(0)-x^*\|^2+\frac{nM\theta_1\eta_0(L\varphi(0)+2)}{\ln \lambda^{-1}}\frac{1-\lambda^t}{\varphi(t)}\\
&~~~+\left(2nM\left(\frac{\theta_2(1-\lambda^t)}{\ln\lambda^{-1}}+\frac{\theta_3\sqrt{1-\lambda^{2t}}}{\sqrt{\ln\lambda^{-2}}}\right)\right.\\
&~~~~~~~~~\left.\frac{mK^2}{2}\right)\frac{1}{\varphi(t)}\int_0^t\eta_s^2ds.
\end{split}
\end{equation}
Let $\eta_t=\frac{\beta}{(t+1)^a}$ with $\frac{1}{2}<a\leq1$. It is straightforward to verify that
$\varphi(t)=\mathcal{O}(t^{1-a})$ if $0<a<1$ and $\varphi(t)=\mathcal{O}(\ln t)$ if $a=1$, which further implies that $1/\varphi(t)=\mathcal{O}(t^{-(1-a)})$ if $0<a<1$ and $1/\varphi(t)=\mathcal{O}((\ln t)^{-1})$ if $a=1$. From (\ref{eqth1}), the upper bound of $\mathbb{E}[f(x_i(t))-f(x^*)]$ consists of several terms with different decay behaviors. Specifically, by Lemma \ref{LE3}, the terms in the first line of (\ref{eqth1}), including $\lambda^t$, $\int_0^t \lambda^{t-\tau}\eta_\tau d\tau$, and $\left(\int_0^t(\lambda^2)^{t-\tau}\eta_\tau^2d\tau\right)^{\frac{1}{2}}$, all converge at the rate $\mathcal{O}(t^{-a})$ if any $a>0$.
Moreover, (\ref{eqth1}) also contains integral terms involving $\varphi(s)\eta_s^2$ and $\varphi(s)^2\eta_s^2$. A direct calculation shows that $\int_0^t \varphi(s)\eta_s^2ds=\mathcal{O}(t^{2-3a})$ if $\frac{1}{2}<a<\frac{2}{3}$, $\mathcal{O}(\ln t)$ if $a=\frac{2}{3}$, and $\mathcal{O}(1)$ if $\frac{2}{3}<a\leq1$, while $\int_0^t \varphi(s)^2\eta_s^2ds=\mathcal{O}(t^{3-4a})$ if $\frac{1}{2}<a<\frac{3}{4}$, $\mathcal{O}(\ln t)$ if $a=\frac{3}{4}$, and $\mathcal{O}(1)$ if $\frac{3}{4}<a\leq1$. Finally, the deterministic terms scaled by $1/\varphi(t)$ in (\ref{eqth1}) converge at the rate $\mathcal{O}(t^{-(1-a)})$ if $0<a<1$ and $\mathcal{O}((\ln t)^{-1})$ if $a=1$. Submitting the bounds to (\ref{eqth1}) immediately implies (\ref{eqco1}). This completes the proof.
\end{proof}

\section{Simulation examples}
In this section, we present a simulation example to illustrate the validity of our result. Consider an MAS consisting of six agents, indexed by the set $\mathcal{V}=\{1,\cdots, 6\}$. Each agent communicates with its neighbors via a time-varying directed graph that periodically switches among four subgraphs, as shown in Fig. \ref{fig1}. The weight of each edge is assumed to be 1. The switching order of the four subgraphs follows $(a)\rightarrow(b)\rightarrow(c)\rightarrow(d)\rightarrow(a)\rightarrow\cdots$, and the holding time of each subgraph is set to $0.01\,\text{s}$. Note that the union of all possible graphs is strongly connected. Then, the connectivity of the graph in Fig. \ref{fig1} satisfies the conditions of Assumption \ref{ASB1} with $T_c=0.04$.

We define the state of agent $i$ as $x_i=[x_{i1},x_{i2}]^T$. For any $i\in \mathcal{V}$, the local objective function is given by
\[f_i(x)=a_ix_1^2-b_ix_1+\frac{i}{6}x_2^2-\frac{i+1}{2}x_2,\]
where $x=[x_1,x_2]^T$ is the decision variable. Suppose that $a_1=0.3$, $a_2=a_3=0.15$, $a_4=a_6=0.1$, $a_5=0.2$, and $b_1=b_3=0.5$, $b_2=b_4=0.8$, $b_5=b_6=0.2$. The sum of the local objective functions is computed as $f(x)=x_1^2-3x_1+\frac{21}{6}x_2^2-\frac{27}{2}x_2$. It is obvious that the global objective function $f(x)$ is convex. By simple computation, we have $x^*=[1.5, 1.93]^T$. In the simulation, we select $\eta_t=\frac{2}{t+1}$ and $g_i(t)=[\sin(t), \cos(t)]^T$. The initial states are set as $x_1(0)=[0.3, 2]^T$, $x_2(0)=[0.5, 1.3]^T$, $x_3(0)=[0.7, 2.7]^T$, $x_4(0)=[0.9, 1]^T$, $x_5(0)=[1.1, 3]^T$ and $x_6(0)=[1.3, 1.6]^T$. By independently running algorithm (\ref{eq3_1}) in 3000 rounds, we plot the averages of the agents' states and the averages of $f(x_i(t))-f(x^*)$ in Fig. \ref{fig2} and Fig. \ref{fig3}, respectively. From Fig. \ref{fig2},  we see that the states of all agents converge asymptotically to the common minimizer $x^*=[1.5, 1.93]^T$ in expectation. While from Fig. \ref{fig3}, we know that after a period of evolution time, the expectation of $f(x_i(t))$ approaches to the optimal value $f(x^*)$. These observations are consistent with the theoretical results established in Theorem \ref{TH1}. Thus, the effectiveness of our method is verified.

\begin{figure}
\centering
\includegraphics[width=0.4\textwidth]{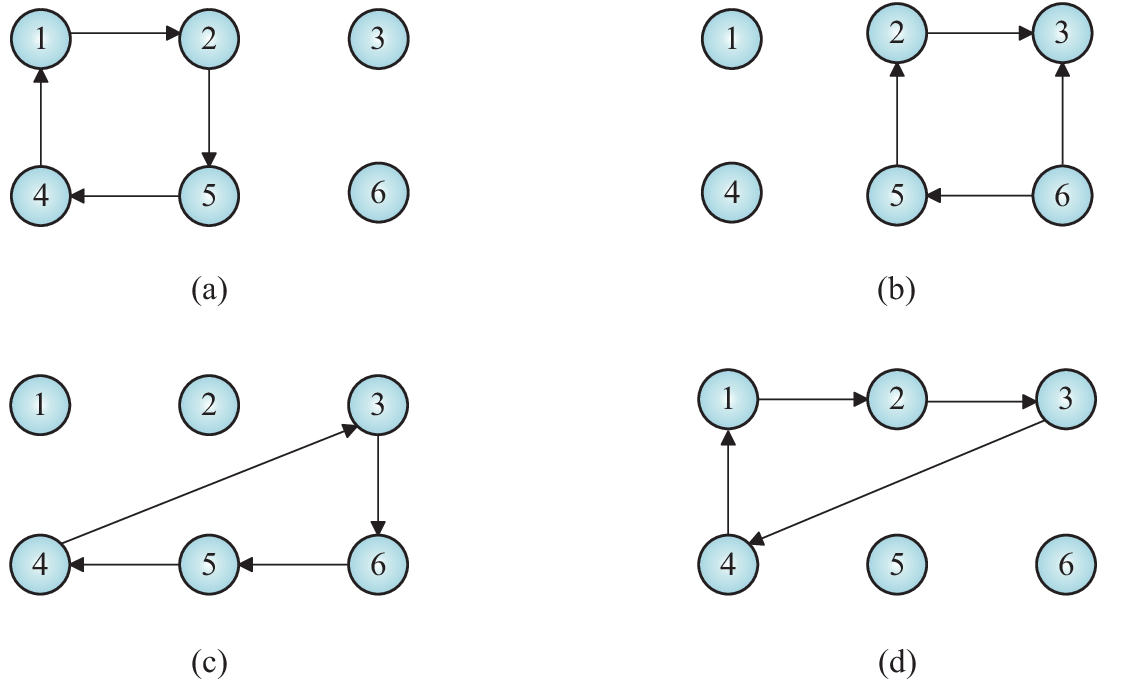}
\caption{The time-varying graph.} \label{fig1}
\end{figure}

\begin{figure}
\centering
\includegraphics[width=0.35\textwidth]{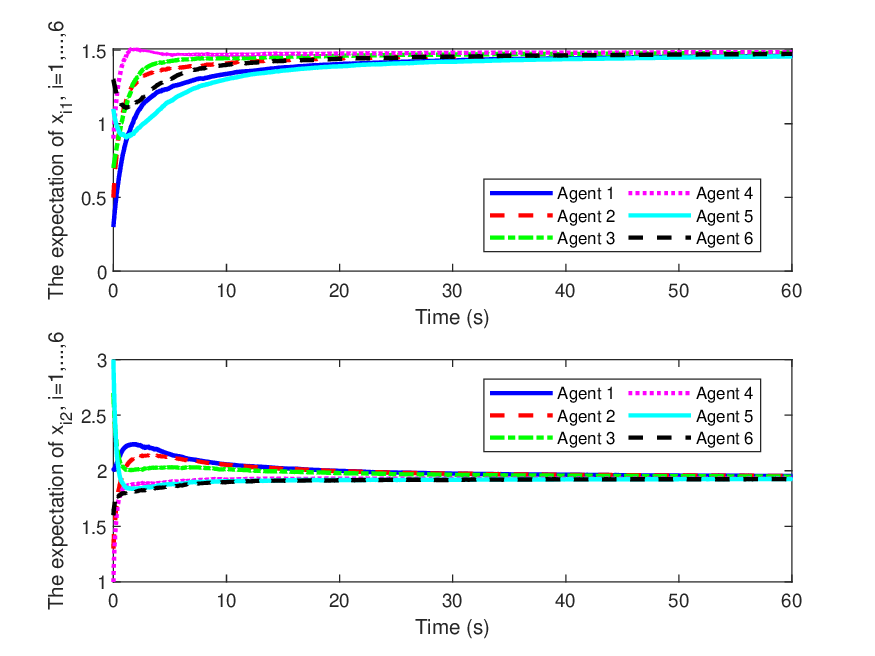}
\caption{The trajectories of $\mathbb{E}[x_i(t)], i=1,\cdots, 6$ under algorithm (\ref{eq3_1}).} \label{fig2}
\end{figure}

\begin{figure}
\centering
\includegraphics[width=0.35\textwidth]{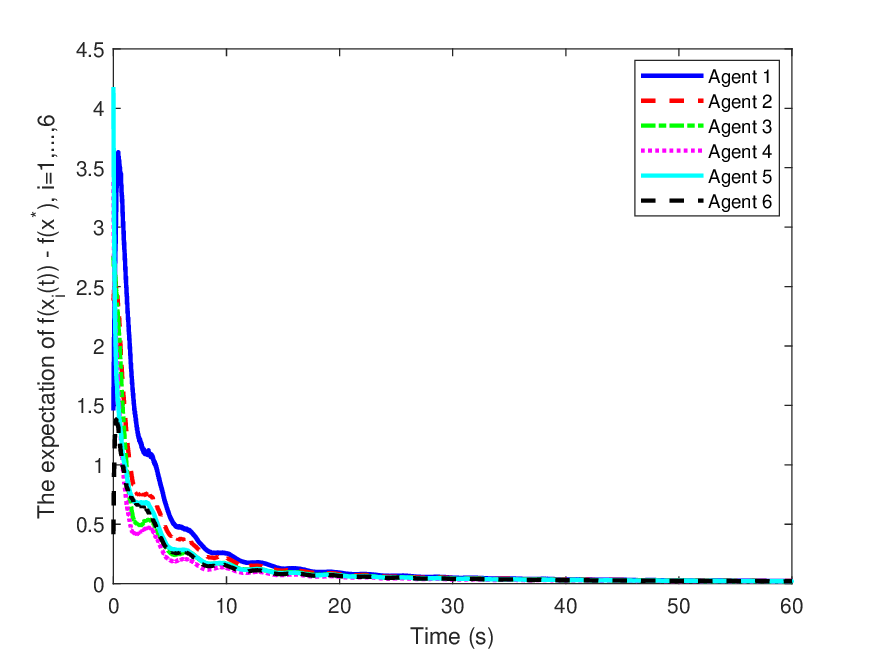}
\caption{The trajectories of $\mathbb{E}[f(x_i(t))-f(x^*)], i=1,...,6$ under algorithm (\ref{eq3_1}).} \label{fig3}
\end{figure}

\section{Conclusions}
In this paper, we have presented a continuous-time distributed stochastic gradient algorithm to address the distributed convex optimization problem. By implementing the algorithm, each agent makes decisions only using the stochastic gradient influenced by the Brownian motion, as well as the local state information. The convergence is analyzed under a time-varying directed communication graph. The results show that if the digraph is $(\delta,T_c)$-strongly connected and balanced, then the proposed algorithm converges in expectation, and the corresponding convergence rate can be as fast as $\mathcal{O}\left(t^{-1/4}\sqrt{\ln t}\right)$. A simulation example has been provided to demonstrate the effectiveness of our theoretical results.
Our future work will focus on some other interesting topics, such as the cases with non-convex objective functions, constraints, and convergence in high probability, which will bring new challenges in studying continuous-time distributed stochastic optimization.

\section{Appendix}
\subsection{Proof of Lemma \ref{LE3}}
For $t\geq t_0$, one has
\begin{equation}\begin{split}\label{eq1271}
&\int_{0}^{t}\frac{\lambda^{-s}}{(s+1)^a}ds\\
&=\frac{\lambda^{-s}}{\ln\lambda^{-1}(s+1)^a} \bigg|_{0}^{t}+\frac{a}{\ln \lambda^{-1}} \int_{0}^{t}\frac{\lambda^{-s}}{(s+1)^{a+1}}ds\\
&\leq\frac{\lambda^{-t}}{\ln \lambda^{-1}(t+1)^a}+ \frac{a}{\ln\lambda^{-1}} \int_{0}^{t_0} \frac{\lambda^{-s}}{(s+1)^{a+1}}ds\\
&~~~+\frac{a}{\ln\lambda^{-1}}\int_{t_0}^{t} \frac{\lambda^{-s}}{(s+1)^{a+1}}ds\\
&\leq\frac{\lambda^{-t}}{\ln \lambda^{-1}(t+1)^a}+ \frac{a \lambda^{-t_0} t_0}{\ln \lambda^{-1}}\\
&~~~+\frac{a}{\ln\lambda^{-1}(t_0+1)}\int_{0}^{t} \frac{\lambda^{-s}}{(s+1)^a}ds.
\end{split}\end{equation}
Due to the fact that $\frac{a}{\ln \lambda^{-1}(t_0+1)}<1$, inequality (\ref{eq1271}) implies that for any $t\geq t_0$,
\begin{equation}\begin{split}\label{eq1272}
\int_{0}^{t}\frac{\lambda^{-s}}{(s+1)^a}ds&\leq \frac{t_0+1}{\ln\lambda^{-1}(t_0+1)-a}\frac{\lambda^{-t}}{(t+1)^a}\\
&+\frac{a\lambda^{-t_0}t_0(t_0+1)}{\ln \lambda^{-1}(t_0+1)-a}.
\end{split}\end{equation}
Moreover, for $t< t_0$, we have $\int_{0}^{t} \frac{\lambda^{-s}}{(s+1)^a}ds\leq \int_{0}^{t_0} \frac{\lambda^{-s}}{(s+1)^a}ds\leq \lambda^{-t_0} t_0$. Together with (\ref{eq1272}), it immediately leads to the validity of the result.

\subsection{Proof of Lemma \ref{LE1}}

Denote $x(t)=\mathrm{col}(x_i(t))_{i\in \mathcal{V}}$, $F(x(t))=\mathrm{col}(\nabla f_i(x_i(t)))_{i\in \mathcal{V}}$, $B(t)=\mathrm{col}(B_i(t))_{i\in \mathcal{V}}$ and $g(t)=\mathrm{diag}(g_i(t))_{i\in \mathcal{V}}$, by algorithm (\ref{eq3_1}), we have
\begin{equation*}\label{eqso1}
\begin{split}
dx(t)=(-(L(t)\otimes I_m)x(t)-\eta_tF(x(t)))dt-\eta_tg(t)dB(t).
\end{split}
\end{equation*}
Based on the result on linear equations in the narrow sense \cite{C16}, we have
\begin{equation}\begin{split}\label{eq3-3}
x(t)=&(\Phi(t,0)\otimes I_m)x(0)-\int_0^t(\Phi(t,s)\otimes I_m)\eta_s F(x(s))ds\\
&-\int_0^t(\Phi(t,s)\otimes I_m)\eta_s g(s)dB(s),
\end{split}\end{equation}
for any $t\geq 0$, where $\Phi(\cdot,\cdot)$ is defined in (\ref{eqb2}). Note that $\Phi(t,s)$ is doubly stochastic for any $t,s\geq 0$, one has
\begin{equation}\begin{split}\label{eq3-4}
\bar{x}(t)=&\left(\frac{\mathbf{1}_n^T}{n}\otimes I_m\right)x(0)-\int_0^t\left(\frac{\mathbf{1}_n^T}{n}\otimes I_m\right)\eta_s F(x(s))ds\\
&-\int_0^t\left(\frac{\mathbf{1}_n^T}{n}\otimes I_m\right)\eta_s g(s)dB(s).
\end{split}\end{equation}
By combining (\ref{eq3-3}) and (\ref{eq3-4}), it follows that
\begin{equation*}\begin{split}
&x_i(t)-\bar{x}(t)\\
&=-\int_0^t\left(\left([\Phi(t,s)]_{i\cdot}-\frac{\mathbf{1}_n^T}{n}\right)\otimes I_m\right)\eta_s F(x(s))ds\\
&~~~-\int_0^t\left(\left([\Phi(t,s)]_{i\cdot}-\frac{\mathbf{1}_n^T}{n}\right)\otimes I_m\right)\eta_s g(s)dB(s)\\
&~~~+\left(\left([\Phi(t,0)]_{i\cdot}-\frac{\mathbf{1}_n^T}{n}\right)\otimes I_m\right)x(0).
\end{split}\end{equation*}
It implies that
\begin{equation}\begin{split}\label{eq3-5}
&\|x_i(t)-\bar{x}(t)\|\\
&\leq\int_0^t\left\|\left(\left([\Phi(t,s)]_{i\cdot}-\frac{\mathbf{1}_n^T}{n}\right)\otimes I_m\right)\eta_s F(x(s))\right\|ds\\
&~~~+\left\|\int_0^t\left(\left([\Phi(t,s)]_{i\cdot}-\frac{\mathbf{1}_n^T}{n}\right)\otimes I_m\right)\eta_s g(s)dB(s)\right\|\\
&~~~+\left\|\left(\left([\Phi(t,0)]_{i\cdot}-\frac{\mathbf{1}_n^T}{n}\right)\otimes I_m\right)x(0)\right\|\\
&\leq\left\|\int_0^t\left(\left([\Phi(t,s)]_{i\cdot}-\frac{\mathbf{1}_n^T}{n}\right)\otimes I_m\right)\eta_s g(s)dB(s)\right\|\\
&~~~+\sqrt{nm}\mathcal{C}\|x(0)\|\lambda^t+\sqrt{m}n^{\frac{3}{2}}\mathcal{C}M\int_0^t\lambda^{t-s}\eta_s ds\\
&=\theta_1\lambda^{t}+\theta_2\int_0^t\lambda^{t-s}\eta_s ds+\|\Psi(t)\|,
\end{split}\end{equation}
where the second inequality results from Lemma \ref{LEb1} and the fact that $\|col(w_i)_{i=1}^n\|\leq\sum_{i=1}^n\|w_i\|$. Based on (\ref{eq3-5}), we have
\begin{equation}\begin{split}\label{eq3-6}
&\int_0^t\phi(s)\|x_i(s)-\bar{x}(s)\|ds\\
&\leq\theta_1\int_0^t\phi(s)\lambda^sds+\theta_2\int_0^t\phi(s)\int_0^s\lambda^{s-\tau}\eta_\tau d\tau ds\\
&~~~+\int_0^t\phi(s)\|\Psi(s)\|ds.
\end{split}\end{equation}
Because $\phi(t)>0$ is non-increasing, it holds that $\int_0^t\phi(s)\lambda^sds\leq\phi(0)\int_0^t\lambda^sds=\frac{\phi(0)}{ln\lambda^{-1}}(1-\lambda^t)$. Furthermore, note that
\begin{equation}\begin{split}\label{eq3-7}
\int_0^t\phi(s)\int_0^s\lambda^{s-\tau}\eta_\tau d\tau ds&=\int_0^t\int_0^s\phi(s)\lambda^{\theta}\eta_{s-\theta} d\theta ds\\
&=\int_0^t\lambda^{\theta}\int_\theta^t\phi(s)\eta_{s-\theta}dsd\theta\\
&\leq\int_0^t\lambda^{\theta}\int_\theta^t\phi(s-\theta)\eta_{s-\theta}dsd\theta\\
&\leq\int_0^t\lambda^{\theta}d\theta\int_0^t\phi(s)\eta_{s}ds\\
&=\frac{1}{\ln\lambda^{-1}}(1-\lambda^t)\int_0^t\phi(s)\eta_sds,
\end{split}\end{equation}
where the first equality holds by the change of variables $s-\tau=\theta$, the second equality is obtained by changing the order of integration, and the first inequality follows from the fact that $\phi(s)$ is non-increasing. Then, combining the inequality $\int_0^t\phi(s)\lambda^sds\leq\frac{\phi(0)}{\ln\lambda^{-1}}(1-\lambda^t)$ and substituting inequalities (\ref{eq3-7}) into (\ref{eq3-6}) immediately imply (\ref{eq3-2}) in Lemma \ref{LE1}.
\subsection{Proof of Lemma \ref{LE2}}

First, it is straightforward to establish the following result using Lemma \ref{LEb1} and Assumption \ref{AS3-1},
\begin{equation*}
\begin{split}
&\left\|\left(\left([\Phi(t,s)]_{i\cdot}-\frac{\mathbf{1}_n^T}{n}\right)\otimes I_m\right)\eta_s g(s)\right\|^2\\
&\leq(\sqrt{nm}\mathcal{C}\lambda^{t-s}\eta_s nK)^2,
\end{split}
\end{equation*}
where $\mathcal{C}$, $\lambda$ are defined in Lemma \ref{LEb1}, and $g(s)$ are defined in the proof of Lemma \ref{LE1}, then
\begin{equation}\label{eq3-h3}
\begin{split}
&\mathbb{E}[\|\Psi(t)\|]\\
&=\mathbb{E}\left[\left\|\int_0^t\left(\left([\Phi(t,s)]_{i\cdot}-\frac{\mathbf{1}_n^T}{n}\right)\otimes I_m\right)\eta_s g(s)dB(s)\right\|\right]\\
&\leq \left(\mathbb{E}\left[\left\|\int_0^t\left(\left([\Phi(t,s)]_{i\cdot}-\frac{\mathbf{1}_n^T}{n}\right)\otimes I_m\right)\eta_s g(s)dB(s)\right\|^2\right]\right)^{\frac{1}{2}}\\
&=\left(\mathbb{E}\left[\int_0^t\left\|\left(\left([\Phi(t,s)]_{i\cdot}-\frac{\mathbf{1}_n^T}{n}\right)\otimes I_m\right)\eta_s g(s)\right\|^2ds\right]\right)^{\frac{1}{2}}\\
&\leq\left(\mathbb{E}\left[\int_0^tn^3m\mathcal{C}^2K^2(\lambda^2)^{t-s}\eta_s^2ds\right]\right)^{\frac{1}{2}}\\
&=\theta_3\left(\int_0^t(\lambda^2)^{t-s}\eta_s^2ds\right)^{\frac{1}{2}},
\end{split}
\end{equation}
where the first inequality is obtained by applying H\"{o}lder's inequality and the second equality follows from Lemma \ref{LEE1}.
Using (\ref{eq3-1}) in Lemma \ref{LE1} and (\ref{eq3-h3}), we have
\begin{equation*}\begin{split}\label{eq3-h2}
&\mathbb{E}[\|x_i(t)-\bar{x}(t)\|]\\
&\leq \theta_1\lambda^t+\theta_2\int_0^t\lambda^{t-s}\eta_s ds+\mathbb{E}[\|\Psi(t)\|]\\
&\leq \theta_1\lambda^t+\theta_2\int_0^t\lambda^{t-s}\eta_s ds+\theta_3\left(\int_0^t(\lambda^2)^{t-s}\eta_s^2ds\right)^{\frac{1}{2}}.
\end{split}\end{equation*}
It follows from (\ref{eq3-2}) in Lemma \ref{LE1} that
\begin{equation}\begin{split}\label{eq3-9}
&\mathbb{E}\left[\int_0^t\phi(s)\|x_i(s)-\bar{x}(s)\|ds\right]\\
&\leq \frac{\theta_1\phi(0)}{\ln\lambda^{-1}}(1-\lambda^t)+\frac{\theta_2}{\ln\lambda^{-1}}(1-\lambda^t)\int_0^t\phi(s)\eta_sds\\
&~~~+\mathbb{E}\left[\int_0^t\phi(s)\|\Psi(s)\|ds\right]\\
&\leq \frac{\theta_1\phi(0)}{\ln\lambda^{-1}}(1-\lambda^t)+\frac{\theta_2}{\ln\lambda^{-1}}(1-\lambda^t)\int_0^t\phi(s)\eta_sds\\
&~~~+\theta_3\int_0^t\phi(s)\left(\int_0^s(\lambda^2)^{s-\tau}\eta_\tau^2d\tau\right)^{\frac{1}{2}}ds\\
&\leq \frac{\theta_1\phi(0)}{\ln\lambda^{-1}}(1-\lambda^t)+\frac{\theta_2}{\ln\lambda^{-1}}(1-\lambda^t)\int_0^t\phi(s)\eta_sds\\
&~~~+\theta_3\left(\int_0^t\phi(s)^2ds \int_0^t\int_0^s(\lambda^2)^{s-\tau}\eta_\tau^2d\tau ds\right)^{\frac{1}{2}},
\end{split}\end{equation}
where the second inequality follows from (\ref{eq3-h3}) and Fubini's theorem \cite{XY}, and the last inequality is true due to Cauchy-Schwarz inequality. Further, using similar methods in (\ref{eq3-7}), we have
\begin{equation}\begin{split}\label{eq3-10}
\int_0^t\int_0^s(\lambda^2)^{s-\tau}\eta_\tau^2d\tau ds\leq\frac{1}{\ln\lambda^{-2}}(1-\lambda^{2t})\int_0^t\eta_s^2ds.
\end{split}\end{equation}
Substituting (\ref{eq3-10}) into (\ref{eq3-9}), we conclude that the result holds.

\end{document}